\documentclass[submission,copyright,creativecommons]{eptcs}
\providecommand{\event}{IWIGP 2012} % Name of the event you are submitting to

\usepackage{algpascal}
\usepackage{amsmath}
\usepackage{amsthm}
\usepackage{cite}
\usepackage{url}
\usepackage{breakurl}

\newtheorem{theorem}{Theorem}

\title{
A Cryptographic Moving-Knife Cake-Cutting Protocol     %The title of this paper
}
\author{
Yoshifumi Manabe
 \institute{
   NTT Communication Science Laboratories, 2-4 Hikaridai, Seika-cho,
  Kyoto 619-0237 Japan}
   \email{manabe.yoshifumi@lab.ntt.co.jp}
\and
  Tatsuaki Okamoto
 \institute{
 NTT Information Sharing Platform Laboratories,
    3-9-11, Midori-cho, Musashino-shi, 180-8585 Japan}
     \email{okamoto.tatsuaki@lab.ntt.co.jp}
}
\def\titlerunning{A Cryptographic Moving-Knife Cake-cutting Protocol}
\def\authorrunning{Y. Manabe \& T. Okamoto}
\begin{document}

\maketitle

\begin{abstract}
This paper proposes a cake-cutting protocol using
cryptography
when the cake is a heterogeneous good that is represented by an interval
 on a 
real line.
Although the Dubins-Spanier moving-knife protocol
 with one knife achieves simple fairness,
 all players must execute the protocol synchronously.
 Thus,  the protocol cannot be executed on asynchronous networks such as the
 Internet.
 We show that the moving-knife protocol
can be executed asynchronously by a discrete protocol
 using a secure auction protocol.
 The number of cuts is $n-1$ where $n$ is the number of players,
 which is the minimum.
 
%  keywords: game-theory, cake-cutting, moving-knife,
%  secure auction 
\end{abstract}

\section{Introduction}

Cake-cutting is an old problem in game theory~\cite{book2,book}.
It can be employed for such purposes as dividing territory of
a conquered island or 
assigning jobs to members of a group.

This paper discusses achieving a moving-knife protocol
using cryptography
in cake-cutting when the cake is a
heterogeneous good that is represented by an interval, $[0,1]$,
on a real line.

The moving-knife protocol is a common technique for
achieving  fair cake-cutting.
The trusted third party (TTP) or one of the players
moves a knife on
the cake. Every player watches the movement
and calls `stop' when the knife comes to some
specific point that is desirable for the  player.
Cake is cut at the points the calls are made.
Many protocols that use one or more knives  
were shown to achieve some desirable property such as exact
division \cite{book2}.

The simplest moving-knife protocol using one knife was
proposed by Dubins and Spanier~\cite{moving}. 
The protocol achieves simple-fairness and it is truthful.

Moving-knife protocols have several
disadvantages.
First, all players must watch the knife movement
simultaneously, thus moving-knife protocols
cannot be executed on networks such as the Internet,
in which transmission
delays cannot be avoided.
In addition, moving knives means cutting the cake
at an infinite number of places, thus it is considered to be
inefficient.

Many discrete protocols have been
 proposed
that achieve simple
 fairness~\cite{simpleo0,simpleo1,simpleo2,simpleo3,simpleo4}.
Several different models were proposed that concern the allowed types of primitives.
The simplest model is just minimizing the number of cuts.
Then,  Robertson-Webb model was proposed~\cite{book}. In the model, 
`cut' and `eval' operations are allowed.
The complexity of
 the protocol is given by the total number of the two operations.

However, the cake-cutting problem when applied to 
the simplest model has not yet been completely solved.
Discrete versions of the Dubins-Spanier moving-knife protocol 
considered in \cite{comsoc,simpleo3} are not truthful.

Cryptography is not commonly used in cake-cutting protocols.
A commitment protocol~\cite{commit} is used in meta-envy-free cake-cutting
protocols~\cite{meta} for multiple parties to declare simultaneously
their respective private values.
Complicated cryptographic protocols have not
been used for cake-cutting protocols so far.

\subsection{Our result}

We show a cryptographic cake-cutting
 protocol that achieves simple fairness with the
 minimum number of cuts.
We use a secure auction protocol that calculates the maximum bid
and the winning player
while hiding the bid of each player.
The protocol output is the same as that of
Dubins-Spanier moving-knife protocol. 
The protocol achieves simple fairness and it is truthful.

\section{Preliminaries}

Throughout the paper, the cake is a heterogeneous good
that is represented by  interval $[0,1]$ on a
real line.
Each player $P_i$ has a utility function, $\mu_i$,
that has the following three properties.
\begin{enumerate}
\item For any interval $X \subseteq [0,1]$ whose size is not empty,
$\mu_i(X)>0$.
      \item For any $X_1$ and $X_2$ such that
$X_1 \cap X_2 = \emptyset$, $\mu_i(X_1 \cup X_2)=\mu_i(X_1) +
\mu_i(X_2)$.
\item  $\mu_i([0,1])=1$.
\end{enumerate}
The tuple of the utility function of $P_i(i=1,\ldots, n)$
is denoted as $(\mu_1, \ldots, \mu_n)$.
Utility functions might differ among players.  
No player has knowledge of the utility of the other players.

An $n$-player cake-cutting protocol, $f$, assigns several
portions of $[0,1]$ to the players such that every portion of $[0,1]$
is assigned to one player.
We denote $f_i(\mu_1, \ldots, \mu_n)$ as the
set of portions assigned to player $P_i$ by $f$,
when the tuple of the utility function is 
$(\mu_1, \ldots, \mu_n)$.

All players are risk-averse, namely they avoid gambling.
They try to maximize the worst case utility they can obtain. 

A desirable property for cake-cutting protocols is
truthfulness.
A protocol is truthful if there is no incentive for any player
to lie about his utility function.
If a player obtains more utility by 
declaring a false value, the protocol is not robust.
For example, consider the simplest cake-cutting protocol
`divide-and-choose.'
In this protocol, Divider first cuts the cake into
two pieces $[0,x]$ and $[x,1]$, such that $\mu([0,x])=\mu([x,1])=1/2$ for
Divider.
Chooser selects the piece she prefers.
Divider obtains the remaining piece.
Since the utility function of Divider is unknown to Chooser,
Divider can lie about his utility function and
cut the cake as  $[0,x']$ and $[x',1]$, for any $x' (\not= x)$.
In this case, Chooser might select the piece such that
 the utility for Divider is more than half
 and Divider might obtain less than half.
 Thus, the risk-averse Divider obeys the rule of the protocol
 and cuts the cake in half.
 `Divide-and-choose' is thus truthful for risk-averse
 players. 

Several desirable properties of cake-cutting protocols 
have been defined~\cite{book}.
Simple fairness, which is the most fundamental one,
is defined as follows.

For any $i$, $\mu_i(f_i(\mu_1, \ldots, \mu_n)) \ge 1/n$.

This paper discusses simple-fair cake-cutting protocols.
One of the other types of the desirable property is the
social surplus, that is, the total utilities
the players obtain.
For two protocol $f$ and $f'$ which has the same
properties (for example, both truthful and simple fair),
$f$ is better than $f'$ in the sense of social surplus if
$\sum_{i=1}^n \mu_i(f_i(\mu_1, \ldots, \mu_n)) > \sum_{i=1}^n
\mu_i(f'_i(\mu_1, \ldots, \mu_n))$.

Several kinds of complexity models
of discrete cake-cutting problems are defined.
The simplest model is that the complexity is the total number of cuts.
This model is further divided into two categories.
\begin{itemize}
\item Cut-and-calculate model: Any operation that uses  the utility
      function of each player is possible other than cutting.
\item Cut-only model: No operation other than cutting is allowed.
      Thus, the utility of player $P_i$ can be known only by $P_i$
      performing a cut. 
\end{itemize}
 
Another model called the Robertson-Webb model is introduced.
The operations are
restricted to the following two types in the model.

\begin{itemize}
\item $Cut_i(I,\alpha)$: Player $P_i$ cuts interval  $I=[x_1,x_2]$
such that $\mu_i([x_1, y])= \alpha \mu_i(I)$, where $0 \le  \alpha \le
1$.

\item $Eval_i(I)$: Player $P_i$ evaluates interval $I=[x_1,x_2]$, which is
one of the cuts previously performed using the protocol.
$P_i$ returns $\mu_i(I)$.
\end{itemize}

The complexity of the Robertson-Webb model is defined as follows.
\begin{itemize}
\item Robertson-Webb cut-complexity model:
      The complexity is measured by the number of cuts.
That is, evaluation queries can be issued for free. 
\item Robertson-Webb cut-and-query-complexity model:
      The complexity is measured by the total
      number of cuts and queries.
\end{itemize}

For the cut-only model, when the number of players is $n=3$,
the minimum number of
cuts for simple-fair division is three~\cite{book}.
When $n=4$, the minimum number of cuts is four~\cite{simpleo0}.
For a general number of players,
the Divide and Conquer protocol~\cite{simpleo0} achieves $1+ nk -2^k$ cuts, where $k=\lfloor
 \log_2 n \rfloor$~\cite{minimal}.
The lower bound of the cut-only model is $\Omega(n \log n)$~\cite{hardness}. 

For the  Robertson-Webb cut-and-query-complexity model, the lower bound
is $\Omega(n \log n)$~\cite{sgall}.
Edmonds and Pruhs extended the $\Omega(n \log n)$ lower bound to
the cases when a player obtains
a union of intervals and approximate fairness is achieved~\cite{edmonds}.

\section{Dubins-Spanier moving-knife protocol}\label{previous}

This section outlines the
Dubins-Spanier moving-knife protocol~\cite{moving} shown in Fig. \ref{DS}.

\begin{figure}[h]
\begin{algorithmic}[1]
\Begin 
\State Let $k \gets n$ and $x \gets 1$.
\Repeat 
\State The TTP moves the knife from $x$ toward
 $0$.
 Let $y$ be the current position of the knife.
\State Player $P_i$ calls `stop' if $\mu_i([y,x])=\mu_i([0,x])/k$.
\State The TTP immediately stops
moving the knife when 'stop' is called.
 Let $x'$ be the point of the knife when `stop' is called.
\State The TTP cuts the cake at $x'$.
The player who said `stop' obtains the piece $[x',x]$
 and exits the protocol.
\State Let $k \gets k-1$ and $x \gets x'$.
\Until{ k=1 }
\State The remaining player obtains the rest of the cake ($[0,x]$).
\End.
\end{algorithmic}
\caption{Dubins-Spanier moving-knife protocol}\label{DS} 
\end{figure}

When the number of remaining players is $k$ and the remaining cake is $[0,x]$,
each remaining player $P_i$ calls `stop' if the knife comes to point $y$
which satisfies $\mu_i([y,x])=\mu_i([0,x])/k$,
that is, the value of piece $[y,x]$ is $1/k$ of the remaining cake.
The first player who calls `stop' obtains piece $[y,x]$
and exits the protocol. The remaining players continue the same
procedure for the remaining cake $[0,y]$.

Each player obtains at least $1/n$ based on the utility function of
the player, thus simple-fairness is achieved.

In addition, the protocol is truthful for risk-averse players.
Consider the case when player $P_i$ tells a lie.
Assume that the number of current remaining players is $k$.
Let the remaining players be 
$P_i, P_{i+1}, \ldots, P_{i+k-1}$ 
and the remaining cake be $[0,x]$.
The actual place that $P_i$ to call `stop' is $x_i$, that is,
$\mu_i([x_i,x])=\mu_i([0,x])/k$.

If $P_i$ calls `stop' earlier than $x_i$,
 $P_i$ obtains less than $\mu_i([0,x])/k$
and the result is worse than telling the truth.

If $P_i$ does not call `stop' even if
the knife comes to $x_i$, 
player $P_{i+1}$ might call `stop' at $x_i-\epsilon$.
The remaining piece is $[0,x_i-\epsilon]$ and
$\mu_i([0,x_i-\epsilon]) < (k-1) \mu_i([0,x])/k$.
Let $x_{i+1}=x_i-\epsilon$.
After that, player $P_j (j=i+2, i+3, \ldots, i+k-1)$
calls `stop' at point $x_j$ such that
$\mu_i([x_j,x_{j-1}])=\mu_i([0,x])/k$.
If $P_i$ calls `stop' before $x_j ( j >i+1)$, $P_i$
obtains less than $\mu_i([0,x])/k$.
If $P_i$ does not call `stop' and
obtains the last remaining piece $[0, x_{i+k-1}]$,
the utility of $P_i$, $\mu_i([0,x_{i+k-1}])$,
is less than $ \mu_i([0,x])/k$.
Therefore, not calling 'stop' at the true point
can be worse than telling the truth.

Note that the moving-knife protocol is not
a discrete protocol.
A protocol is presented by Endriss~\cite{comsoc}
shown in Fig. \ref{max} that makes the protocol discrete.

\begin{figure}[h]
\begin{algorithmic}[1]
\Begin 
\State Let $k \gets n$ and $x \gets 1$.
\Repeat 
\State Each player $P_i$ declares point $x_i$ such that 
$\mu_i([x_i,x])= \mu_i([0,x])/k$.
\State Let $x'$ be the maximum of $x_i$s. Let $P_i$ be the
 player who called  $x'$.
\State $P_i$ obtains piece $[x',x]$ and exits the protocol.
\State Let $k \gets k-1$ and $x \gets x'$.
\Until{ k=1 }
\State The remaining player obtains the rest of the cake ($[0,x]$).
\End.
\end{algorithmic}
\caption{Endriss protocol}\label{max} 
\end{figure}

It seems that this protocol is the same as the 
Dubins-Spanier moving-knife protocol, but it is actually not.
In this protocol, all players know the cut
point of the other players.
The cut point information can offer a hint to a player
and the player can obtain more utility
by behaving dishonestly.
Suppose that $k=3$ and the 
density functions for the utility of the players are as follows.
\begin{equation*}
u_1(z)=
\begin{cases}
 4/5  & 0 \le z \le 5/6 \\
 2    & 5/6 < z \le 1 
\end{cases}
\end{equation*}
$$u_2(z)=1( 0 \le z \le 1),$$
\begin{equation*}
u_3(z)=
\begin{cases}
 2   & 0 \le z \le 1/3 \\
 1/2   & 1/3 < z \le 1 
\end{cases}
\end{equation*}
The utility of $P_i$ for $[x,y]$, $\mu_i([x,y])$,
is calculated by $\int^y_x u_i(z) dz$.
Since $\int^1_0 u_i(z) dz=1 (i=1,2,3)$, these density
functions satisfy the conditions of the utility functions.

At the first round, each player declares
$c_1=5/6$, $c_2=2/3$, and $c_3=1/3$,
since $\int^1_{5/6} u_1(z) dz=1/3$, 
$\int^1_{2/3} u_2(z) dz=1/3$, and
$\int^1_{1/3} u_3(z) dz=1/3$. 
Since $5/6>2/3>1/3$, $P_1$ obtains $[5/6,1]$ and exits the protocol.
The next round is performed  by $P_2$ and $P_3$
with the remaining cake $[0,5/6]$.
The honest declaration, $c'_2$, at the next round
by $P_2$ is $5/12$, since
$\int^{5/6}_{5/12} u_2(z) dz=1/2 \int^{5/6}_0
u_2(z) dz =5/12$.
Since $\int^{5/6}_{11/48} u_3(z) dz=1/2 \int^{5/6}_0
u_3(z) dz$, $P_3$ will declare $11/48$ as the
cut point $c_3'$, for the next round.

Although $P_2$ cannot know $c_3'$ in advance,
it knows that $c_3'< c_3$ is satisfied for any utility function.
Thus, $P_2$ can declare a false value $1/3(=c_3)$, instead of the
true value of $5/12$ as $c_2'$, if $P_2$ knows that
 the declared value by $P_3$ in
previous round is $c_3$.
When $P_2$ declares false value $1/3$,
$P_2$ wins in this round and obtains $[1/3, 5/6]$.
The utility of $P_2$ is $1/2$, which is larger than utility
$5/12$ when $P_2$ declares the true cut point, $5/12$.

Thus knowledge of  the declared values of other players
destroys the truthful characteristic of the protocol.
The trimming protocol~\cite{simpleo3}, which also achieves simple-fair
by a discrete protocol, has the same problem about truthfulness, since a
player might be able to know all other players' cut points in the
previous round. 

Sgall and Woeginger showed a protocol in which 
the number of cuts is $n-1$, shown in Fig. \ref{sgall}.
 
 \begin{figure}[h]
\begin{algorithmic}[1]
\Begin 
\State Each player, $P_i$, simultaneously
declares $n-1$ points $x_{i,j}( 1\le j \le n-1)$
such that $\mu_i([x_{i,j}, x_{i,j+1}])=1/n ( 0 \le j \le n-1)$
 (Note that $x_{i,0}=0$ and $x_{i,n}=1$).
\State Let $y \gets 0$. 
\For{k=1}{n-1}
\Begin
\State Let $z \gets \min x_{i,k}$, where the minimum is taken
among the remaining players. 
\State Let $P_j$ be the player who declares $z$.
\State $P_j$ obtains $[y,z]$ and exits the protocol.
\State Let $y \gets z$.
\End
\State The remaining player obtains the rest of the cake ($[y,1]$).
\End.
\end{algorithmic}
\caption{Sgall-Woeginger protocol}\label{sgall} 
\end{figure}

This protocol achieves simple fairness.
When $k=1$, player $P_i$ who obtains  piece $[0,z]$
satisfies $z=x_{i,1}$, thus $\mu_i([0,x_{i,1}])=1/n$.
Next consider the case $k>1$.
If player $P_i$ obtains $[y,x_{i,k}]$ in the $k$-th round,
$P_i$ could not obtain its piece in the previous round.
Thus, $y \le x_{i,k-1}$ is satisfied for for any currently 
remaining player $P_i$ at line 6 
and $\mu_i([y,x_{i,k}]) \ge \mu_i([x_{i,k-1},x_{i,k}])=1/n$.

Since all players declare their cut points simultaneously, 
no player can know the other players' cut points in advance.
Thus, telling a false value such as in the Endriss protocol
is not effective in this protocol.

The assignment result differs from the one of
original Dubins-Spanier moving-knife protocol.
In the moving-knife protocol,
when $P_i$ exits in the first round with obtaining $[x,1]$,
each of the remaining player $P_j$
obtains at least $\mu_j([0,x])/(n-1)$, which is greater than
$1/n$.
Since $P_j$ did not win in the first round,
$\mu_j([x,1])< 1/n$, thus $\mu_j([0,x]) > (n-1)/n$.
Therefore, from the second round,
the cake is more than $(n-1)/n$ for the remaining players.
The other rounds have the same characteristic.
If a player exits with a ``small''(in the other players' view)
portion of the cake, all of the remaining players
obtains more utility.

On the other hand, in the Sgall-Woeginger protocol,
when a player exits with a ``small'' portion of the cake,
the extra part of the cake is automatically assigned
to the next round's winner.
For example, $P_i$ wins in the first round and obtains
$[0,x]$ and exits,
remaining player $P_j$ thinks that the remaining
cake is $(n-1)/n+ \mu_j([x,x_{j,1}])$, where $\mu_j([0,x_{j,1}])=1/n$.
In the next round, the player $P_k$ wins whose $x_{k,2}$ is smallest
among the remaining players, but the value of the extra part
$\mu_k([x,x_{k,1}])$ might not large among the remaining players.

In Dubins-Spanier moving-knife protocol, next round call
 is done for all of the remaining cake, thus the extra part
(such as  $[x,x_{j,1}]$) is also considered by the remaining players.
Next round winner is the player  who values the highest to
the extra part of the remaining cake.
The next round winner is satisfied with a relatively `small' portion of
the cake because of the extra part, thus
the next round remaining cake can be larger
than in the Sgall-Woeginger protocol.
Thus, in the view of the social surplus, the Dubins-Spanier
 moving knife is more desirable than Sgall-Woeginger protocol.

 \section{Cryptographic moving-knife protocol}

The important characteristics of the Dubins-Spanier moving-knife
protocol are
that (1) the declaration is done round by round and
(2) when a player $P$ calls `stop', no player
knows the other remaining players' cut points 
because the knife is moving so that the size
of the cutting piece increases.

Because of the first characteristic,
the social surplus is better than Sgall-Woeginger protocol.
Because of the second characteristic,
every player does not know the previous round
cut point information of the other remaining players.

The simplest solution to keep the protocol truthful and make the
protocol discrete would be to have a TTP.
In each round, every remaining  player privately sends its cut point
to the TTP.
The TTP decides the largest value and the player
who gave the maximum value from the cut point information.

However, it might be difficult to have
such a TTP. There might be collusion between a player and the TTP.
The TTP might send the player
cut point information to the colluding player.

In order to address this problem,
we introduce a secure auction protocol.
Secure auction protocols have been proposed
in cryptography theory~\cite{auction3,auction,auction2}.
They are outlined as follows.
\begin{itemize} 
\item Player $P_i$ generates its share of public key and secret key,
      $(PK_i, SK_i)$ of a homomorphic encryption scheme. 

      $P_i$ broadcasts $PK_i$ and the public encryption key $PK$
      is calculated by any player from $(PK_1, \ldots, PK_n)$.

      $SK_i$ is the private key of $P_i$ for decryption.

Any player can execute encryption procedure $Enc$ using $PK$.
      The ciphertext obtained by executing $Enc$ on plaintext $m$
      is $Enc(PK,m)$.

If $P_1, \ldots, P_n$ jointly execute decryption procedure $Dec$
      with their private keys $SK_1, \ldots, SK_n$,
      they can decrypt $Enc(PK,m)$ and obtain $m$.
      That is, $Dec(Enc(PK,m), SK_1, \ldots, SK_n)=m$.
      Note that the decryption can be performed without revealing the value
      of $SK_i$
      to any other players.

      For any set of players whose size is less than $n$,
      they cannot decrypt $Enc(PK,m)$ by themselves.

\item $P_i$ encrypts his bid $b_i$ using the public key,
      that is, $P_i$ calculates $c_i=Enc(PK,b_i)$.

\item $P_1, \ldots, P_n$ jointly calculates $b_{max}=\max(b_1, \ldots, b_n)$
      and player $P_j$ who bids $b_{max}$
      from $c_1, \ldots, c_n$
      without directly decrypting $c_1, \ldots, c_n$ using the
      homomorphic property.
      
\item During execution of the secure auction protocol,
      each player gives a zero-knowledge proof\cite{ZKP}
      that the player acts correctly.
The proof can be verified by any other player.

      The correctness of the obtained highest bid
      and the winner player is also given as a zero-knowledge proof.
      The proof can be verified by any player.
      That is, no player can deny its bid afterwards.
\end{itemize}

The details are shown in \cite{auction3,auction,auction2}.
Secure auction protocols use a homomorphic encryption, in which
addition of encrypted values can be accomplished
 without
 decrypting them.
Homomorphic encryption has the following properties.
\begin{itemize}
\item There exists polynomial time computable operation $\otimes$ and $^{-1}$ as follows.
      For any two ciphertext $c_1 =Enc(PK,m_1)$ and $c_2 =Enc(PK,m_2)$,
      $c_1 \otimes c_2 \in  Enc(PK, m_1+m_2)$.

      For any ciphertext $c =Enc(PK,m)$, $c^{-1} \in Enc(PK, -m)$. 
\item The encryption is semantically secure, that is, the advantage of
      the adversary for 
      the following game is negligible.

      The adversary obtains all $PK_i$'s and all $SK_i$'s except for
      some $SK_j$. First, the adversary
      can repeatedly obtain $Dec(SK,c)$ for any
      ciphertext $c$ that it selects. It then outputs two plaintext $m_0,m_1$.
      Challenger randomly selects bit $b \leftarrow \{0,1\}$ and
      $c=Enc(PK, m_b)$ is given to the adversary.

      Then the adversary outputs $b'$. It wins if $b=b'$

      The advantage of the adversary is $Pr[b=b']-1/2$.
 \end{itemize}
The first property is calculating sum of two ciphertexts without
decrypting them.
 Using the homomorphic characteristics, it is possible to
 compare multiple bids without decrypting them, that is,
 they can obtain $C=Enc(PK,\max(b_1,\ldots, b_n))$ from $c_1, \ldots, c_n$.
 They jointly decrypt $C$ and obtain the maximum bid without knowing
 each bid.
In some secure auction protocol~\cite{auction2},
another type of homomorphic encryption scheme is used in which 
multiplication of two ciphertexts are also possible.

The second  property means that no player can obtain information
of the plaintext from a given ciphertext if  at least one of
the secret keys is unknown. 
 
The moving-knife protocol
using a secure auction protocol
is shown in Fig. \ref{cmk}.
In auction protocols, the bids are considered to be an integer.
Thus, we convert cake $[0,1]$ to $[0,2^m]$ for some
large integer $m$ and each player must bid an integer value for the
cutting point. Note that $m$ must be large enough such that
for any player $P_i$ and any $c \in$ [0,1],
$\mu_i([\lfloor c \cdot 2^m \rfloor/2^m, c])$
is negligible, that is, bidding integer values
is not a bad  approximation.

\begin{figure}[h]
\begin{algorithmic}[1]
\Begin
\State Let $k \gets n$, $x \gets 2^m$.
\Repeat
\State $P_i$ decides $x_i$ such that
$\mu_i([x_i,x])=\mu_i([0,x])/k$.
\State $P_i$ encrypts $x_i$ and broadcasts it.
\State All players execute a secure auction protocol together
 and obtain  maximum bid $c$ and player $P$ who bids $c$.
\State $[c,x]$  is marked as the piece for $P$ and $P$ cannot bid any more.
\State Let $x \gets c$, $k \gets k-1$.
\Until{k=1}.
\State $[0,x]$ is marked as the piece for the remaining player  and
every player obtains his/her piece.
\End. 
\end{algorithmic}
\caption{Cryptographic moving-knife protocol.}\label{cmk} 
\end{figure}
This protocol achieves simple fairness.
The protocol is asynchronous, that is, 
 no two events in this protocol
need to be executed simultaneously.
The number of cuts is $n-1$, which is the minimum.

A difference between the Dubins-Spanier moving-knife protocol and
this protocol is
  that no player exits the protocol during the execution.
If a player exits, the set of players who
execute the secure auction protocol changes in each round.
Changing the set of players requires that the keys be re-generated 
for the secure auction protocol, thus the protocol would be inefficient.
Therefore, the set of players is unchanged in this protocol.
However, if a player obtains a piece, the player has no incentive to
execute the secure auction protocol honestly any more.
Thus, in the proposed protocol, the pieces are actually assigned to the players
at the end of the protocol.
During the execution of the secure auction protocol, each player presents 
a proof that the player executes the protocol correctly. 
If a player misbehaves, it is detected by verifying the
  proof and
the player does not obtain the piece marked for the player.
This assignment at the end of the protocol
must also be done without TTP.
If this protocol is executed just once, there is
no way to prevent a player from misbehaving.
If this protocol is executed multiple times
or some other protocol will be executed among the
same players, there is a record of the  proof that
a player misbehaved in this execution of the protocol,
and the player will be rejected from joining another protocol or
another execution of this protocol.
If a player wants not to be rejected, the player has an incentive to
act correctly.

\begin{theorem}
The protocol in Fig. \ref{cmk} is truthful for
 risk-averse players  and simple fair.
The number of cuts is minimum. 
\end{theorem}
\begin{proof}
These properties are achieved because
 the assignment is exactly the same as the Dubins-Spanier
 moving-knife protocol. 
\end{proof}

\section{Conclusion}
This paper proposed a cryptographic cake-cutting protocol.
The protocol is discrete and truthful
It achieves simple fairness with the minimum number of cuts.

Further study will include 
the use of cryptography in other cake-cutting protocols.
 
{\bf Acknowledgment} We thank Dr. Hiro Ito and anonymous referees 
for their valuable comments.

\bibliographystyle{eptcs}
\bibliography{cut}

\end{document}